%% file: grossi-hastewaste_new.tex
\documentclass[proceedings]{stacs}
\stacsheading{2009}{517--528}{Freiburg}
\firstpageno{517}

\usepackage{amsthm,xspace}
\usepackage{amsmath,amssymb}
\stacsheading{year}{numbers}{city}
\def\ProofInAppendix{TRUE}
\usepackage{proofapnd}
\usepackage{bm}

\newcommand\rank{\ensuremath{\mathtt{rank}}}
\newcommand\select{\ensuremath{\mathtt{select}}}
\newcommand\rankz{\ensuremath{\mathtt{rank}_0}}
\newcommand\ranko{\ensuremath{\mathtt{rank}_1}}
\newcommand\selectz[1]{\ensuremath{\mathtt{select}^{#1}_0}}
\newcommand\selecto[1]{\ensuremath{\mathtt{select}^{#1}_1}}
\newcommand\poly[1]{\mathrm{poly}(#1)}
\newcommand\polylog[1]{\mathrm{polylog}(#1)}

\newcommand\eps{\varepsilon}
\newcommand\FID{\textsc{fid}}
\newcommand\PRED{\textsc{pred}}
\newcommand\VEB{\textsc{veb}}

\begin{document}
\title[Lowering the Redundancy in Fully Indexable Dictionaries]{More Haste, Less Waste: Lowering the Redundancy \\ in Fully Indexable Dictionaries}
\author[unipi]{R. Grossi}{Roberto Grossi}
\address[unipi]{Dipartimento di Informatica, Universit\`a di Pisa, Italy}
\email{{grossi,aorlandi}@di.unipi.it}

\author[unipi]{A. Orlandi}{Alessio Orlandi}
\author[rram]{R. Raman}{Rajeev Raman}
\address[rram]{Department of Computer Science, University of Leicester,
 United Kingdom}
\email{r.raman@mcs.le.ac.uk}
\author[ssrao]{S. S. Rao}{S. Srinivasa Rao}
\address[ssrao]{\textsc{madalgo} Center$^{*}$\footnote{$^{*}$Center for Massive Data Algorithmics, a center of the Danish National Research Foundation}, Aarhus University, Denmark}
\email{ssrao@daimi.au.dk}


\def\today{}

\begin{abstract}
  We consider the problem of representing, in a compressed format, a
  bit-vector~$S$ of $m$ bits with $n$ $\mathbf{1}$s, supporting the
  following operations, where $b \in \{ \mathbf{0}, \mathbf{1} \}$:
\begin{itemize}
\item $\mathtt{rank}_b(S,i)$ returns the number of occurrences of bit
  $b$ in the prefix $S\left[1..i\right]$;
\item $\mathtt{select}_b(S,i)$ returns the position of the $i$th
  occurrence of bit $b$ in $S$.
\end{itemize}
Such a data structure is called \emph{fully indexable dictionary
  (\FID)} [Raman, Raman, and Rao, 2007], and is at least as powerful
as predecessor data structures.  Viewing $S$ as a set $X = \{ x_1,
x_2, \ldots, x_n \}$ of $n$ distinct integers drawn from a universe
$[m] = \{1, \ldots, m\}$, the predecessor of integer $y \in [m]$ in
$X$ is given by $\selecto{}(S, \ranko(S,y-1))$. {\FID}s have many
applications in succinct and compressed data structures, as they are
often involved in the construction of succinct representation for a
variety of abstract data types.

\smallskip

Our focus is on space-efficient {\FID}s on the \textsc{ram} model with word size
$\Theta(\lg m)$ and constant time for all operations, so that
the time cost is independent of the input size.

Given the bitstring $S$ to be encoded, having length $m$ and
containing $n$ ones, the minimal amount of information that needs to
be stored is $B(n,m) = \lceil \log {{m}\choose{n}} \rceil$.  
The state of the art in building a \FID\ for~$S$ is given in~\mbox{}[P\v{a}tra\c{s}cu, 2008] using $B(m,n)+O( m / ( (\log m/ t) ^t) ) + O(m^{3/4}) $ bits, to support the operations in $O(t)$ time.

Here, we propose a parametric data structure exhibiting a time/space
trade-off such that, for any real constants $0 < \delta \leq 1/2$, $0
< \eps \leq 1$, and integer $s > 0$, it uses
\[
B(n,m) + O\left(n^{1+\delta} + n \left(\frac{m}{n^s}\right)^\eps\right)
\]
bits and performs all the operations in time $O(s\delta^{-1} +
\eps^{-1})$.  The improvement is twofold: our redundancy can be
lowered parametrically and, fixing $s = O(1)$, we get a constant-time
\FID\ whose space is $B(n,m) + O(m^\eps/\poly{n})$ bits, for
sufficiently large $m$.  This is a significant improvement compared to
the previous  bounds for the general case.

\end{abstract}

\maketitle
\vspace*{-5ex}
\section{Introduction}
\label{sec:introduction}
\input{s1-intro}

\section{Elias-Fano Revisited}
\label{sec:elias-fano-revisited}
\input{s2-eliasfano}

\section{Basic Components and Main Result}
\label{sec:multiranking-recursion-polynomial-squeezing}
\input{s3-basic}

\section{Main Ideas for Achieving Polynomial Redundancy}
\input{s4-polynomial}


%
%
\smallskip
\noindent {\bf Acknowledgements.} The first two authors would like to thank Sebastiano Vigna for precious
discussion. Thanks also go to the anonymous referees for useful comments. Work partially supported by the MAINSTREAM Italian MIUR Project.
\bibliography{refs}
\bibliographystyle{plain}
\vspace*{-1.2cm}
\end{document}

%% file: s1-intro.tex

Data structures for dictionaries \cite{BM99,munro96,pagh01,RRR07},
text indexing \cite{CM96,FM05,GV05,HSS03,MRR01,NM07}, and representing
semi-structured data \cite{FLMM05,GRRR06,GRR06,MR01,RRR07}, often
require the very space-efficient representation of a bivector $S$ of
$m$ bits with $n$ $\mathbf{1}$s (and $m-n$ $\mathbf{0}$s).  Since
there are ${m \choose n}$ possible choices of $n$ $\mathbf{1}$s out of
the $m$ bits in $S$, a simple information-theoretic argument shows
that we need at least $B(n,m) = \lceil \log {m \choose n} \rceil$ bits
of space, in the worst case, to store $S$ in some compressed format.
However, for the aforementioned applications, it is not enough just to
store the compressed $S$, as one would like to support the following
operations on $S$, for $b \in \{ \mathbf{0}, \mathbf{1} \}$:
\begin{itemize}
\item $\mathtt{rank}_b(S,i)$ returns the number of occurrences of bit
  $b$ in the prefix $S\left[1..i\right]$;
\item $\mathtt{select}_b(S,i)$ returns the position of the $i$th
  occurrence of bit $b$ in $S$.
\end{itemize}

Our focus will be on space-efficient data structures that 
support these operations efficiently, on the \textsc{ram} model with word size 
$\Theta(\log m)$.  The resulting data structure is called a
\emph{fully indexable dictionary} (\FID) \cite{RRR07} and is quite
powerful. For example, $S$ can equally represent a set $X = \{ x_1,
x_2, \ldots, x_n \}$ of $n$ distinct integers drawn from a universe
$[m] = \{1, \ldots, m \}$, where $S\left[x_i\right] = \mathbf{1}$, for
$1 \leq i \leq n$, while the remaining $m-n$ bits of $S$ are
$\mathbf{0}$s. In this context, the classical problem of finding the
\emph{predecessor} in $X$ of a given integer $y \in [m]$ (i.e.\mbox{}
the greatest lower bound of $y$ in $X$) can be solved with two {\FID}
queries on $S$ by $\selecto{}(S, \ranko(S,y-1))$.  {\FID}s have also
connections with coding theory, since they represent a sort of locally
decodable source code for $S$ \cite{Buhrman:2002:BO}. They are at the
heart of compressed text indexing since they enable space to be
squeezed down to the high-order entropy when properly employed
\cite{GGV03}.  Finally, they are the building blocks for many complex
low space data structures \cite{BB08,FGGSV08,munro08,MRRR03} that
require $O(1)$ lookup time, namely, their time complexity is
independent of the number of entries stored at the expense of using
some extra space.

To support the $\mathtt{rank}$ and $\mathtt{select}$ operations in
$O(t)$ time, for some parameter~$t$, it appears to be necessary to use
additional space, beyond the bound $B(n,m)$ needed for representing
the bitstring $S$ in compressed format.  This extra space
is termed the \emph{redundancy} $R(n,m,t)$ of the data structure, and
gives a total of $B(n,m) + R(n,m,t)$ bits \cite{GM07}.  Although the
leading term $B(n,m)$ is optimal from the information-theoretic point
of view, a discrepancy between theory and practice emerges when
implementing {\FID}s for various applications
\cite{CN08,GGMN05,GGV04,GHSV07,OS07,Vigna08}.  In particular, the term
$B(n,m)$ is often of the same order as, if not superseded by, the
redundancy term $R(n,m,t)$. For example, consider a constant-time
\FID\ storing $n=o(m/\polylog{m})$ integers from the universe $[m]$:
here, $B(n,m)$ is negligible when compared to the best known bound
of $R(n,m,1) = O(m / \polylog{m})$ \cite{P08}.

Our goal is that of reducing the redundancy $R(n,m,t)$ for the general
case $n \leq m$.  Although most of the previous work has generally
focussed on the case $t = O(1)$, and $m = n \cdot \polylog{n}$, the
burgeoning range of applications (and their complexity) warrant a much
more thorough study of the function $R(n,m,t)$.  

There are some inherent limitations on how small can the redundancy
$R(n,m,t)$ be, since {\FID}s are connected to data structures for the
predecessor problem, and we can inherit the predecessor lower bounds
regarding several time/space tradeoffs.  The connection between
{\FID}s and the predecessor problem is well
known~\cite{BF02,GHSV07,PT06,RRR07} and is further developed in this
paper, going beyond the simple inheritance of lower bounds.  A
predecessor data structure which gives access to the underlying data
set is, informally, a way to support \emph{half} the operations
natively: either $\selecto{}$ and $\ranko$, or $\selectz{}$ and
$\rankz{}$.  In fact, we show that a data structure solving the
predecessor problem can be turned into a \FID\ and can also be made to
store the data set using $B(n,m)+O(n)$ bits, under certain assumptions
over the data structure.

Consequently, if we wish to understand the limitations in reducing the
redundancy $R(n,m,t)$ of the space bounds for {\FID}s, we must briefly
survey the state of the art for the lower bounds involving the
predecessor problem.  The work in \cite{PT06} shows a number of lower
bounds and matching upper bounds for the predecessor problem, using
data structures occupying at least $\Omega(n)$ words, from which we
obtain, for example, that $R(n,m,1)$ can be $o(n)$ only when $n =
\polylog{m}$ (a degenerate case) or $m = n \,\polylog{n}$.  For $m =
n^{O(1)}$, the lower bound for $B(n,m) + R(n,m,1)$ is
$\Omega(n^{1+\delta})$ for any fixed constant $\delta>0$.  Note that
in the latter case, $B(n,m) = O(n \log m) = o(R(n,m,1))$, so the
``redundancy'' is larger than $B(n,m)$.  Since $\ranko$ is at least
as hard as the predecessor problem, as noted in \cite{BF02,PT06}, 
then all \FID{}s suffer from the same limitations. (It
is obvious that $\rankz$ and $\ranko$ have the same complexity, as
$\rankz(S,i) + \ranko(S,i) = i$.) As noted in \cite[Lemma 7.3]{RRR07},
$\selectz{}$ is also at least as hard as the predecessor problem.
Other lower bounds on the redundancy were given for ``systematic''
encodings of~$S$ (see \cite{GM07,golynski07,M05} and related papers),
but they are not relevant here since our focus is on
``non-systematic'' encodings \cite{GGGRR07,GRR08}, which have provably
lower redundancy. (In ``non-systematic'' encodings one can store $S$
in compressed format.)

In terms of upper bounds for $R(n,m,t)$, a number are known,
of which we only enumerate the most relevant here. 
For systematic structures,
an optimal upper bound is given by~\cite{golynski07}
for $R(n,m,O(1)) = O(m \log \log m / \log m)$.
Otherwise, a very recent upper
bound in \cite{P08} gives $R(n,m,t) = O(m / ((\log m)/t)^t + m^{3/4}
\polylog{m})$ for any constant $t > 0$. 
These bounds are most interesting when
$m = n \cdot \polylog{n}$.  As noted earlier, sets that are sparser
are worthy of closer study. For such sets, one cannot have best of two
worlds: one would either have to look to support queries in
non-constant time but smaller space, or give up on attaining $R(n,m,
1) = o(B(n,m))$ for constant-time operations.

The main role of generic case \FID s is expressed when they take part in more
structured data structures (e.g. succinct trees)
where there is no prior knowledge of the relationship between $n$ and $m$.
Our main contribution goes along this path, striving for constant-time operations.
Namely, we devise a
constant-time \FID\ having redundancy $R(n,m,O(1)) = O(n^{1+\delta} +
n(m/n^s)^\eps)$, for any fixed constants $\delta < 1/2$, $\eps < 1$
and $s>0$ (Theorem~\ref{the:multirank_recursion}).  The running time
of the operations is always $O(1)$ for $\selecto{}$ (which is
insensitive to time-space tradeoffs) and is $O(\eps^{-1} +
s\delta^{-1}) = O(1)$ for the remaining operations.  When $m$ is
sufficiently large, our constant-time \FID\ uses just $B(n,m) +
O(m^\eps/\poly{n})$ bits, which is a significant improvement compared
to the previous bounds for the general case, as we move 
from a redundancy of kind $O(m/\polylog
{m})$ to a one of kind $O(m^\eps)$, by proving for the first time that
polynomial reduction in space is possible.

Moreover, when instantiated in a polynomial universe case (when $m =
\Theta(n^{O(1)})$, for a sufficiently small $\eps$, the redundancy is dominated
by $n^{1+\delta}$, thus extending the known predecessor search data structure
with all four \FID\   operations
without using a second copy of the data.
Otherwise, the $m^\eps$ term is dominant when the universe is superpolynomial,
e.g. when $m = \Theta(2^{\log^c n})$ for $c > 1$.
In such cases we may not match the lower bounds for predecessor search;
however, this is the price for a solution which is agnostic of $m,n$ relationship.


We base our findings on the Elias-Fano encoding scheme
\cite{Elias74,Fano71}, which gives the basis for {\FID}s naturally
supporting $\selecto{}$ in $O(1)$ time.  


%% file: s2-eliasfano.tex

We review how the Elias-Fano scheme \cite{Elias74,Fano71,
  OS07,Vigna08} works for an arbitrary set $X = \{ x_1 < \cdots < x_n
\}$ of $n$ integers chosen from a universe $[m]$.  Recall that $X$ is
equivalent to its characteristic function mapped to a bitstring $S$ of
length $m$, so that $S\left[x_i\right] = \bm{1}$ for $1 \leq i \leq n$
while the remaining $m-n$ bits of $S$ are $\bm{0}$s.  Based on the
Elias-Fano encoding, we will describe the main ideas behind our new
implementation of fully indexable dictionaries (\FID s).  We also
assume that $n \leq m/2$---otherwise we build a \FID\ on the
complement set of $X$ (and still provide the same functionalities),
which improves space consumption although it does not guarantee
$\selecto{}$ in $O(1)$ time.

\paragraph*{\bf Elias-Fano encoding.}  
\label{sub:elias-fano-encoding}
Let us arrange the integers of $X$ as a \emph{sorted} sequence of
consecutive words of $\log m$ bits each. Consider the first\footnote{Here we use Elias' original choice of ceiling and floors, thus our bounds slightly differ from the \emph{sdarray} structure of~\cite{OS07}, where they obtain $n\lceil\log(m/n)\rceil + 2n$.} $\lceil \log n\rceil$
bits of each integer $x_i$, called $h_i$, where $1 \leq i \leq n$.
We say that any
two integers $x_i$ and $x_j$ belong to the same \emph{superblock} if
$h_i = h_j$.

The sequence $h_1 \leq h_2 \leq \cdots \leq h_n$ can be stored as a
bitvector $H$ in $3n$ bits, instead of using the standard $n \lceil
\log n\rceil$ bits. It is the classical unary representation, in which
an integer $x \geq 0$ is represented with $x$ $\bm{0}$s followed by a
$\bm{1}$. Namely, the values $h_1, h_2-h_1, \ldots, h_n-h_{n-1}$ are
stored in unary as a multiset. 
For example, the sequence $h_1, h_2, h_3, h_4, h_5 =
1,1,2,3,3$ is stored as $H = \bm{01101011}$, where the $i$th $\bm{1}$
in $H$ corresponds to $h_i$, and the number of $\bm{0}$s from the
beginning of $H$ up to the $i$th $\bm{1}$ gives $h_i$ itself.  The
remaining portion of the original sequence, that is, the last $\log m
- \lceil \log n\rceil$ bits in $x_i$ that are not in $h_i$, are stored
as the $i$th entry of a simple array $L$. Hence, we can reconstruct
$x_i$ as the concatenation of $h_i$ and $L\left[i\right]$, for $1 \leq
i \leq n$.  The total space used by $H$ is at most $2^{\lceil \log n
  \rceil} + n \le 3n$ bits and that used by $L$ is $n \times (\log m -
\lceil \log n\rceil) \leq n \log (m/n)$ bits.

Interestingly, the plain storage of the bits in $L$ is related to the
information-theoretic minimum, namely, $n \log (m/n) \leq B(n,m)$
bits, since for $n \leq m/2$,  $B(n,m) \sim n\log(m/n) + 1.44\,n$ by means of Stirling
approximation.  
In other words, the simple way of representing the
integers in $X$ using Elias-Fano encoding requires at most $n \log
(m/n) + 3n$ bits, which is nearly $1.56\,n$ away from the theoretical
lower bound $B(n,m)$. 
If we employ a constant-time \FID\ to store $H$,
Elias-Fano encoding uses a total of $B(n,m) + 1.56\,n + o(n)$ bits.

\paragraph*{\bf Rank and select operations vs predecessor search.}
\label{sub:rank-select-vs-predecessor}
Using the available machinery---the \FID\ on $H$ and the plain array
$L$---we can perform $\select_1(i)$ on $X$ in $O(1)$ time: we first recover
$h_i = \select_1(H,i)-i$ and then concatenate it to the fixed-length
$L\left[i\right]$ to obtain $x_i$ in $O(1)$ time \cite{GV05}.
As for $\rank$ and $\select_0$, we point out that they are 
intimately related to the \emph{predecessor search},
as we show below (the converse has already been pointed out in the
Introduction).


Answering $\ranko(k)$ in $X$ is equivalent to finding the predecessor
$x_i$ of $k$ in $X$, since $\ranko(k) = i$ when $x_i$ is the
predecessor of $k$. Note that $\rankz(k) = k - \ranko(k)$, so
performing this operation also amounts to finding the predecessor. As
for $\selectz{}(i)$ in $X$, let $\overline{X} = [m] \setminus X = \{
v_1, v_2, \dots, v_{m-n} \}$ be the
complement of $X$, where $v_i < v_{i+1}$, for $1 \le i < m-n$. 
Given any $1 \leq i \leq m-n$, our goal is to find $\selectz{}(i) =
v_i$ in constant time, thus motivating that our assumption $n \leq
m/2$ is w.l.o.g.: whenever $n \leq m/2$, we store the complement set
of $X$ and swap the zero- and one-related operations.

The key observation comes from the fact that we can associate each
$x_l$ with a new value $y_l = \bigl|\{ v_j \in \overline{X} \mathrm{\
 such\ that\ } v_j < x_l\}\bigr|$, which is the number of elements in
$\overline{X}$ that precede $x_l$, where $1 \leq l \leq n$.  The
relation among the two quantities is simple, namely, $y_l = x_l - l$,
as we know that exactly $l-1$ elements of $X$ precede $x_l$ and so
the remaining elements that precede $x_l$ must originate from
$\overline{X}$. Since we will often refer to it, we call the set
$Y = \{ y_1, y_2, \ldots, y_n \}$ the \emph{dual
 representation} of the set $X$.

Returning to the main problem of answering $\selectz{}(i)$ in $X$,
our first step is to find the predecessor $y_j$ of $i$ in $Y$, namely,
the largest index $j$ such that $y_j < i$. As a result, we infer that
$x_j$ is the predecessor of the \emph{unknown} $v_i$ (which will be
our answer) in the set $X$. We now have all
the ingredients to deduce the value of $v_i$. Specifically, the
$y_j$th element of $\overline{X}$ occurs before $x_j$ in the universe,
and there is a nonempty run of elements of $X$ up to and including
position $x_j$, followed by $i-y_j$ elements of $\overline{X}$
up to and including (the unknown) $v_i$. Hence, $v_i = x_j + i
- y_j$ and, since $y_j = x_j - j$, we return $v_i = x_j + i - x_j + j=
i + j$.  (An alternative way to see $v_i = i+j$ is that $x_1, x_2,
\ldots, x_j$ are the only elements of $X$ to the left of the unknown
$v_i$.)  We have thus proved the following.

\begin{lemma}
 \label{lemma:rank-select-predecessor}
 Using the Elias-Fano encoding, the $\selecto{}$ operation takes
 constant time, while the $\rank$ and $\selectz{}$ operations can be
 reduced in constant time to predecessor search in the sets $X$
 and $Y$, respectively.
\end{lemma}

The following theorem implies that we can use both lower and upper
bounds of the predecessor problem to obtain a \FID, and vice versa.
Below, we call a data structure storing $X$ \emph{set-preserving} if it
stores $x_1, \dots, x_n$ \emph{verbatim} in a contiguous set of memory cells.

\begin{theorem}
 \label{the:rank-select-predecessor}
 For a given set $X$ of $n$ integers over the universe $\left[m\right]$,
 let $\FID(t,s)$ be a \FID\ that takes $t$ time and $s$ bits of space  
 to support $\rank$ and $\select$. Also, let $\PRED(t,s)$ be a static
 data structure that takes $t$ time and $s$ bits of space to support
 predecessor queries on $X$, where the integers in $X$ are stored in
 sorted order using $n \log m \leq s$ bits. Then,
 \begin{enumerate}
 \item given a $\FID(t,s)$, we can obtain a $\PRED( O(t), s )$;
 \item given a set-preserving $\PRED(t,s)$, we can obtain a $\FID( O(t),
   s - n \log n + O(n) )$ (equivalently, $R(n,m,t) = s
   - n \log m + O(n)$) with constant-time $\selecto{}$.
 \item if there exists a non set-preserving $\PRED(t,s)$, we can obtain a
   $\FID( O(t), 2 s + O(n) )$ with constant-time $\selecto{}$.
\end{enumerate}
\end{theorem}

\begin{proof}[Proof (sketch).]
 The first statement easily follows by observing that the predecessor
 of $k$ in $X$ is returned in $O(1)$ time by $\selecto{}(S,
 \ranko(S,k-1))$, where $S$ is the characteristic bitstring of $X$.  
 Focusing on the second statement, it suffices to encode $X$ using
 the Elias Fano encoding, achieving space $s - n\log n + O(n)$.
  
 To further support $\selectz{}$, we exploit the properties of $Y$ and $X$.
 Namely, there exists a maximal subset $X' \subseteq X$ so that its dual representation
 $Y'$ is strictly increasing, thus being searchable by a predecessor data structure.
 Hence we split $X$ into $X'$ and the remaining subsequence $X''$ and produce two
 Elias-Fano encodings which can be easily combined by means of an extra $O(n)$ bits
 {\FID } in order to perform $\selecto{}$, $\ranko$ and $\rankz$.
 $\selectz{}$ can be supported by exploiting the set preserviness of the data structure,
 thus building only the extra data structure to search $Y'$ and not storing $Y'$.  When data structures are not set-preserving, we simply replicate the
 data and store $Y'$, thus giving a justification to the $O()$ factor. 
\end{proof}
%


%% file: s3-basic.tex
We now address and solve two questions, which are fundamental to
attain a $O(t)$-time \FID\ with $B(n,m) + R(n,m,t)$ bits of storage
using Lemma~\ref{lemma:rank-select-predecessor} and
Theorem~\ref{the:rank-select-predecessor}: (1)~how to devise an
efficient index data structure that can implement predecessor search
using Elias-Fano representation with tunable time-space tradeoff, and
(2)~how to keep its redundancy $R(n,m,t)$ small.

Before answering the above questions, we give an overview of the two
basic tools that are adopted in our construction (the string
B-tree~\cite{FG99} and a modified van Emde Boas
tree~\cite{PT06,MST::BoasKZ1977}). We next develop our major ideas
that, combined with these tools, achieve the desired
time-space tradeoff, proving our main result.

\begin{theorem}
\label{the:multirank_recursion}
Let $s > 0$ be an integer and let $0 \leq \eps, \delta \leq 1$ be
reals. For any bitstring $S$, $|S| = m$, having cardinality $n$, 
there exists a fully indexable dictionary solving
all operations in time $O(s\delta^{-1} + \eps^{-1})$ using
$B(n,m) + O(n^{1+\delta} + n (m/n^s)^\eps)$ bits of space.
\end{theorem}

\paragraph*{\bf Modified van Emde Boas trees.}
\label{sub:redundancy-patrascu-thorup}
P\v{a}tra\c{s}cu and Thorup~\cite{PT06} have given some matching upper
and lower bounds for the predecessor problem. The discussion
hereafter regards the second branch of their bound: as a candidate
bound they involve the equation (with our terminolgy and assuming our
word RAM model) $t = \log ( \log(m/n) / \log ( z/n ) )$, where $t$ is
our desired time bound and $z$ is the space in bits.  By
reversing the equation and setting $\epsilon = 2^{-t}$, we obtain $z =
\Theta(n(m/n)^\epsilon)$ bits.  As mentioned in \cite{PT06}, the
tradeoff is tight for a polynomial universe $m = n^{\gamma}$, for
$\gamma > 1$, so the above redundancy cannot be lower than
$\Theta(n^{1+\delta})$ for any fixed $\delta>0$.

They also describe a variation of van Emde Boas (\VEB) trees
\cite{MST::BoasKZ1977} matching the bound for polynomial universes,
namely producing a data structure supporting predecessor search that
takes $O(\log \frac{\log (m/n)}{\log(z/n)}))$ time occupying $O(z \log
m)$ bits.  In other words, for constant-time queries, we should have
$\log (m/n) \sim \log(z/n)$, which implies that the space is $z =
\Theta(n(m/n)^\epsilon)$.  They target the use of their data structure
for polynomial universes, since for different cases they build
different data structures. However, the construction makes no
assumption on the above relation and we can extend the result to
arbitrary values of $m$. By Theorem~\ref{the:rank-select-predecessor},
we can derive a constant-time \FID\ with redundancy $R(n,m, O(1)) =
O(n(m/n)^\epsilon)$. 

\begin{corollary}
 \label{corollary:VEB-PT}
 Using a modified \VEB\ tree, we can implement a \FID\ that uses
 $B(n,m) + O(n(m/n)^\eps)$ bits of space, and supports all operations
 in $O(\log(1/\eps)$) time, for any constant $\eps > 0$.
\end{corollary}

The above corollary implies that we can obtain a first polynomial
reduction by a straightforward application of existing
results. However, we will show that we can do better for sufficiently
large $m$, and effectively reduce the term $n(m/n)^\eps$ to $n^{1+\delta} +
n(m/n^s)^\eps$. The rest of the paper is devoted to this
goal.

\paragraph*{\bf String B-Tree: blind search for the integers.}
\label{sub:sbtree}
We introduce a variant of string B-tree to support predecessor search 
in a set of integers. Given a set of integers $X = \{ x_1, \dots, x_p\}$ from the universe
$[u]$, we want obtain a space-efficient representation of $X$ that supports predecessor
queries efficiently. We develop the following structure:

\begin{lemma}\label{lem:sbt}
Given a set $X$ of $p$ integers from the universe $[u]$, there exists a 
representation that uses extra $O(p\log\log u)$ bits apart from storing the elements of $X$, 
that supports predecessor queries on $X$ in $O(\log p/\log\log u)$ time. The algorithm requires 
access to a precomputed table of size $O(u^{\gamma})$ bits, for some positive constant 
$\gamma < 1$, which can be shared among all instances of the structure with the same universe size.
\end{lemma}

\begin{proof}
The structure is essentially a succinct version of string B-tree on the elements of $X$ 
interpreted as binary strings of length $\log u$, with branching factor $b = O(\sqrt{\log u})$. 
Thus, it is enough to describe how to support predecessor queries in a set of $b$ elements 
in constant time, and the query time follows, as the height of the tree is $O(\log p/\log\log u)$.
Given a set ${x_{1}, x_{2}, \dots, x_{b}}$ of integers from $[u]$ that need to be stored at a node 
of the string B-tree, we construct a compact trie (Patricia trie) over these integers (interpreted
as binary strings of length $\log u$), having $b$ leaves and $b-1$ internal 
nodes. The leaves disposition follows the sorting order of $X$.
Each internal node is associated with a {\it skip value}, indicating the string depth
at which the LCP with previous string ends. Canonically, left-pointing edges are labeled with a {\bf 0}
and right-pointing with a {\bf 1}.
Apart from storing the keys in sorted order, it is enough to 
store the tree structure and the skip values of the edges. This information can be 
represented using $O(b \log\log u)$ bits, as each skip value is at most $\log u$ and the trie is represented in $O(b)$ bits.

Given an element $y \in [u]$, the search for the predecessor of $y$ proceeds in two stages.
In the first stage, we simply follow the compact trie matching the appropriate bits of $y$ to find
a leaf $v$. Let $x_{i}$ be the element associated with leaf $v$. One can show that $x_{i}$ is 
the key that shares the longest common prefix with $y$ among all the keys in $X$. 
In the second stage, we compare $y$ with $x_{i}$ to find the longest common prefix of 
$y$ and $x_{i}$ (which is either the leftmost or rightmost leaf of the internal node at which 
the search ends).  By following the path in the compact trie governed by this longest 
common prefix, one can find the predecessor of $y$ in $X$.
We refer the reader to \cite{FG99} for more details and the correctness of the search algorithm.
The first stage of the search does not need to look at any of the elements associated 
with the leaves. Thus this step can be performed using a precomputed table of size 
$O(u^{\gamma})$ bits, for some positive constant $\gamma < 1$ (by dividing the binary 
representation of $y$ into chunks of size smaller than $\gamma \log u$ bits each). 
In the second stage, finding the longest common prefix of $y$ and $x_{i}$ can be done 
using bitwise operations. We again use the precomputed table to follow the path 
governed by the longest common prefix, to find the predecessor of $y$.
\end{proof}

%% file: s4-polynomial.tex
In this section, we give a full explanation of the main result,
Theorem~\ref{the:multirank_recursion}. We first give an overview, and
then detail the multiranking problem by illustrating remaining details
involving the construction of our data structure.

\subsection{Overview of our recursive dictionary}
\label{sub:framework}

We consider the $\ranko$ operation only, leaving the effective
development of the details to the next sections. 
A widely used approach to the \FID\  problem (e.g. see~\cite{jacobson-thesis,munro96}) lies in splitting the universe $[m]$ into different
chunks and operating independently in each chunk, storing the rank at
the beginning of the block. Queries are redirected into a chunk via
a preliminary \emph{distributing} data structure and the \emph{local}
data structure is used to solve it. Thus, the space occupancy
is the distributing structure (once) plus all chunks. 
Our approach is orthogonal, and it guarantees better control of the
parameter of subproblems we instantiate with respect to many previous approaches.

Let $X$ ($|X| = n$) be the integer sequence of values drawn from $[m]$
and let $q \in [m]$ be a generic rank query.
Our goal is to produce a simple function $f: [m] \to [m/n]$ and a machinery 
that generates a sequence $\tilde X$ from
$X$ of length $n$ coming from the universe $[m/n]$, so that
given the predecessor of $\tilde q = f(q)$ in $\tilde X$, we can
recover the predecessor of $q$ in $X$. By this way, we can
reduce recursively, multiple times, the rank problem while keeping
a single sequence per step, instead of having one data structure per
chunk. 

Easily enough, $f$ is the ``cutting'' operation of the upper 
$\log n$ bits operated by the Elias Fano construction, which generates
$p$ different superblocks.
Let $X^l_1, \ldots, X^l_p$ the sets of lower $\log(m/n)$ bits of values in $X$,
one per superblock. We define our $\tilde X$ as $\tilde X = \cup_{1 \leq i \leq p} X^l_i$, that is, the set of unique values we can extract from the $X^l$s.
Suppose we have an oracle function $\psi$, so that given a value  $\tilde x\in
\tilde X$ and an index $j \in [p]$, $\psi(j, \tilde x)$ is the predecessor of
$\tilde x$ in $X^l_j$. 
We also recall from Section 2 that the upper bit vector $H$ of the Elias
Fano construction over $X$ can answer the query $\ranko ( x / 2^{\lceil \log n
\rceil} )$ in constant time (by performing $\selectz{}(H,  x / 2^{\lceil \log n
\rceil} )$. That is, it can give the rank value at the
beginning of each superblock. 

Given a query $q$ we can perform $\ranko(q)$ in the following way: 
we use $H$ to reduce the problem within the superblock and know the
rank at the beginning of the superblock $j$. We then have the lower bits of
our query ($f(q)$) and the sequence $\tilde X$: we rank $f(q)$ there, obtaining
a certain result, say $v$;
we finally refer to our oracle to find the predecessor of $v$ into $X^l_j$,
and thus find the real answer for $\ranko(q)$.
The main justification of this architecture is the following: in any superblock,
the predecessor of some value can exhibit only certain values in its
lower bits (those in $\tilde X$), thus once given the predecessor of $f(q)$ 
our necessary step is only to reduce the problem within $[|\tilde X|]$
as the lower bits for any superblock are a subset of $\tilde X$.
The impact of such choice is, as explained later, to let us implement
the above oracle in just $O(n^{1+\delta})$ bits, for any $0  < \delta < 1$.
That is, by using a superlinear number of bits in $n$, we will be able
to let $m$ drop polynomially both in $n$ and $m$.

The above construction, thus, requires one to write $X$ in an Elias
Fano dictionary, plus the oracle space and the space to solve the predecessor
problem on $\tilde X$. 
The first part accounts for $B(n,m) + O(n)$ bits, to which we add $O(n^{1+\delta})$ bits for the oracle. By carefully employing the String B-tree we can
shrink the number of elements of $\tilde X$ to $O(n/\log^2 n)$ elements,
leaving us with the problem of ranking on a sequence of such length
and universe $[m/n]$. We solve the problem by replicating the entire schema
from the beginning. Up to the final stage of recursion, the series
representing the space occupancy gives approximately $O( (n\log (m/n))/\log^{2i} n + (n/\log^{2i} n)^{1+\delta})$ bits at the $i$-th step, descending geometrically. Each step can be traversed in constant time during a query, so
the overall time is constant again. More interestingly, at each step we reduce
the universe size of the outcoming sequence to $m n^{-i}$. Thus, 
at the final step $s$, we employ the previous result of Corollary~\ref{corollary:VEB-PT} and obtain a final redundancy of $O(m^\eps n^{1-s\eps})$.

\subsection{Multiranking}
\label{sub:multiranking}

We now give further details on our construction. Mainly, we show that using our choice on
how to build $\tilde X$ and the function $f$, being able to rank over $\tilde X$
we can build the oracle in $O(n^{1+\delta})$ bits. We do it by illustrating, in a broader
framework, the multiranking problem.

We are given a \emph{universe} $\left[u\right]$ (in
our dictionary case, we start by setting $u = m$), and a set of nonempty sequences 
$A_1, \ldots,
A_c$  each containing
a sorted subset of $\left[u\right]$. 
We also define $r = \sum_{1 \leq j \leq n}
|A_j|$ as the global number of elements.
The goal is, given two values $1 \leq i \leq
c$ (the wanted superblock $\hat s$) and $1 \leq q \leq u$ (the
query $f(q)$), perform $\ranko(q)$ in the set $A_i$ (in our case,
the head in $\hat s$ that is predecessor of the searched key) in
$O(1)$ time and small space. 

A trivial solution to this problem would essentially build a \FID\  for
each of the sequences, thus spending a space proportional to $O(c u)$,
which is prohibitive. Instead, we can carefully exploit the global
nature of this task and solve it in less space. The core
of this technique is the \emph{universe scaling} procedure.
We perform the union of all the $A$ sequences and extract a new, single
sequence $\Lambda$ containing only the distinct values that appear in
the union (that is, we kill duplicates). $\Lambda$ is named the 
\emph{alphabet} for our problem and we denote its length with $t \leq r$.
Next, we rewrite all sequences by using rank of their elements
in the alphabet instead of the initial arguments: now each 
sequence is defined on $\left[t\right]$.

The multiranking problem is solved in two phases. We first perform 
ranking of the query $q$ on $\Lambda$ and then we exploit the
information to recover the predecessor in the given set. Here
we achieve our goal to \emph{(i)} decouple a phase that depends
on the universe from one that dependes on the elements and \emph{(ii)} 
have only one version of the problem standing on the initial universe.
The following lemma solves the multiranking problem completely, that is,
outside our original distinction between a oracle and the alphabet
ranking:

\begin{lemma}
\label{lem:multirank2}
There exists a data structure solving the \emph{\texttt{multirank}}
problem over $c$ nonempty increasing sequences 
$\mathbb{A} = \{A_1, \ldots, A_c\}$ with elements
drawn from the universe $\left[u\right]$, having $r$ elements in total
using $B(r,u) + O(r^{1+\delta}) + o(u)$ bits
for any given $0 < \delta < 1/2$.
\end{lemma}
\begin{proof}
Let $\Lambda$ be the alphabet defined over $u$ by the sequences in $\mathbb{A}$, 
and let $t = |\Lambda|$.
For each of the sequences in $\mathbb{A}$ we create a bitvector $\beta_i$ 
of length $t$ where the $\beta_{ij} = \bm{1}$ if $\Lambda_j \in A_i$.
We first view $\beta_i$s as rows of a matrix of size $tc$; since
$t \leq r$ and each of the sequences are non-empty (and hence $r \geq c$), the matrix is of size
$O(r^2)$. We linearize the matrix by concatenating its rows
and obtain a new bitvector $\beta'$ on which we want to perform predecessor
search. We note that the universe size of this bitvector is $O(r^2)$,
that is, the universe is polynomial.
We store $\beta'$ using the data structure of Corollary~\ref{corollary:VEB-PT} 
setting the time to $\log (1/\delta)$, so that space turns out to be
$O(r^{1+\delta})$. 
Finally, we store we store a \FID\ occupying $B(r,u)+o(u)$ that
represents the subset $\Lambda$ of the universe $\left[u\right]$.

Solving the multirank is easy now: given a query $q$ and a set index
$i$, we use the $o(u)$ {\FID } and find $\lambda = \ranko(q)$ in $U$,
which leads to the predecessor into the alphabet $\Lambda$ of our
query $q$.  Since $\lambda \in \left[t\right]$ we can now use the $\beta$ \FID\
to find $p = \ranko( ti + \lambda)$. The final answer is clearly $p -
\ranko(ti)$.
\end{proof}
\subsection{Completing the puzzle}
\label{sub:puzzle}
The multiranking problem is closely connected with the Elias-Fano
representation of Section~\ref{sub:elias-fano-encoding}.
When plugged in our framework, as explained in Section~\ref{sub:framework},
that we can use our data structure itself to implement the ranking procedure.
Similarly we can use it for $\selectz{}$ by employing another set of data.

We are left with just one major detail. Each time we produce the output sequence
$\tilde X$, containing the lower bits for all elements, our only clue for the
number of elements is the worst case upper bound $n$, which is unacceptable.
We now review the whole construction and employ the string B-tree to have
a polylogarithmic reduction on the number of elements, paying $O(n\log\log m)$ 
bits per recursion step.
Generally, at each step we receive a sequence $X_i$ as input and must output 
a new sequence $X_{i+1}$ plus some data structures that can link the predecessor
problem for $X_i$ to $X_{i+1}$. Each $X_i$ is stored in an Elias-Fano dictionary,
and the sets of superblocks and lower bits sequences are built as explained before.
We then apply a further reduction step on the problem cardinality. Each superblock
can be either \emph{slim} or \emph{fat} depending on whether it contains less than $\log^2 n$
elements or not. Each superblock is split into blocks of size $\log^2 n$, apart
from the last block, and for each block we store a String B-tree with fan-out $\sqrt{\log n}$.
Since the block is polylogarithmic in size, by means of shared precomputed tables
we can perform predecessor search in constant time. Slim superblocks are handled
directly by the tree and they do not participate further in the construction.
For each block in a fat superblock, we logically extract its \emph{head}, that is,
the smallest element in it. We now use heads in the multiranking problems and
we build the output sequence $X_{i+1}$ using only heads lower bits. As
there can only be at most $O(n/\log^2 n)$  blocks in fat superblocks, the size of the 
output sequence is at most $O(n/\log^2 n)$.
The oracle is built as usual, on the heads, using $O(n^{1+\delta})$ bits.

Ranking now performs the following steps: for each recursive step, it uses
the Elias-Fano $H$ vector to move into a superblock and at the same
time check if it is slim or fat. In the latter case, it first
outsources the query for the lower bits to the next dictionary, then
feeds the answer to the multiranking instance and returns the actual answer.
Thus, we just proved the following (with $v = \log^2 n$ and $w = n$):

\begin{theorem}
\label{thm:inductiveStep}
Let $w$ and $v$ be two integer parameters and let $0 < \delta < 1/2$ be
a real constant.  Given $X_i, n_i \ge v$ and $m_i > w$, 
where $n_i \leq m_i$, there exists a procedure that produces a
data structure involved in predecessor search. 
The data structure 
occupies $B(n_i,m_i) + O(w + n_i \log \log m_i + n_i^{1+\delta})$
space, and in 
$O(\delta^{-1})$ time, it reduces a predecessor query on
$X_i$ to a predecessor query on 
a new sequence
$X_{i+1}$ of length $n_{i+1} = O(n_i / v)$ over a universe
$\left[m_{i+1}\right]$, where $m_{i+1} = m_i / w$.
\end{theorem}

We must then deal with the last two steps. The first step aims at
supporting $\selectz{}$ since the above data structure can only support
$\ranko{}$. The second step deals with how treat the final sequence
after a number of iteration steps have been executed. We can finally give
the proof of our main result:

\begin{proof}[Proof of Theorem \ref{the:multirank_recursion}]
    Let $X \subseteq [m]$ be the set whose characteristic vector is $S$.
    The data structure involves recursive instances of
    Theorem~\ref{thm:inductiveStep}, by starting with $X_0 = X$ and
    using each step's output as input for the next step. As previously
    mentioned, we must only cover the base case and the last recursive
    step.  We begin by describing the whole data structure, moving to
    algorithms later on. We start by partitioning $X$ into $X'$ and $X''$
    as described in the proof of
    Theorem~\ref{the:rank-select-predecessor}, so that the
    construction is operated on both $X'$ and $X''$.  We now describe
    representation of $X'$; $X''$ is stored in a similar way.  
    We recursively build smaller sequences by invoking
    Theorem~\ref{thm:inductiveStep} exactly $s$ times, using $\delta$
    as given, and parameters $w = n$, $v = \log^2 m$.  By invoking
    Corollary~\ref{corollary:VEB-PT} the space bound easily follows.
    To support $\selectz{}$ on the original sequence, we operate on
    the $X'$ sequence alone, since when transformed to its dual $Y'$, we
    obtain a strictly monotone sequence.
    Interpreting $X'$ as 
    an implicit representation of $Y'$ we
    build a multiset representation for the high bits ($H'$), a new
    set of succinct string B-trees using the superblocks of the dual sequence
    and thought of as operating on $Y'$ (similarly to Theorem~\ref{the:rank-select-predecessor}) and a 
    new set of $s$ recursive applications of Theorem~\ref{thm:inductiveStep}.

    $\selecto{}$ is trivial, thanks to the machinery of 
    Theorem~\ref{the:rank-select-predecessor}.  The $\ranko$
    algorithm for a query $q$ is performed on both $X'$ and $X''$
    \FID: we start by querying $H_0$, the upper bits of $F'_0$
    ($F''_0$ respectively) for $q / 2^{\lceil\log n\rceil}$, thus
    identifying a certain superblock in which the predecessor for
    $q$ can appear. Unless the superblock is slim (refer to proof of
    Theorem~\ref{thm:inductiveStep}) we must continue to search through
	the next lower-order bits.  
	This is done via multiranking, which recurses in a cascading manner
    with the same technique on the $s$ steps up to the last \FID, that
    returns the answer.  The chain is then walked backwards to find the root \FID\
    representative.  We finally proceed through the succinct string
    B-tree to find the head and the next succinct string B-tree until
    we find the predecessor of $q$.  The last step for recursion
    takes $O(\eps^{-1})$ time. All the middle steps for multiranking
    and succinct string B-tree traversals take $O(s\delta^{-1} + s)$
    time. 
    To support $\selectz{}$, we act on $X'$, using exactly the same algorithm
    as before using, but with the collection of data structures built for the dual
    representation $Y'$, and following the steps of
    Theorem~\ref{the:rank-select-predecessor}. 

During the buildup of the recursive process, say being at step $i$,
the size $n'_i$ for sequence $X'_i$ ($i > 1$),  is
upper bounded by $n / \log^{2i} m$, while the universe has size $m /
n^i$. If at any step $2 \leq j \leq s$ the condition $m_j < w = n$ does
not apply, we cannot apply Theorem~\ref{thm:inductiveStep},
so we truncate recursion and use a $o(w)$ \FID\ to store
the sequence $X_j$. This contributes a negligible amount to the redundancy.
We name the \FID\ for each step $F_1$ up to $F_s$. Suppose we 
can recurse for $s$ steps with Theorem~\ref{thm:inductiveStep}, we end
up with a sequence over a universe $m_s = m/n^s$.  
By using Corollary~\ref{corollary:VEB-PT} the space bound is no 
less than $O(n (m/n^s)^\eps)$. 
The $B(n_i,m_i) + O(n_i^{1+\delta})$ factors decrease geometrically,
so the root dominates and we can show that, apart from lower order terms,
the space bound is as claimed.
Otherwise, the total space
$s(n_i, m_i)$ of the recursive data structure satisfies:
\[
s(n_i, m_i) = s(n_{i+1}, m_{i+1}) + \mathrm{space(\FID\ for\ high\
 bits)} + \mathrm{space(string\ B\mbox{-}trees)} + O(n_i^{1+\delta})
\]
where $n_{i+1} = n_i /\log^2 m$ and $m_{i+1} = m_i /n$. 
The claimed redundancy follows easily.
\end{proof}



%% file: grossi-hastewaste_new.bbl
\begin{thebibliography}{10}

\bibitem{BF02}
P.~Beame and F.E. Fich.
\newblock Optimal bounds for the predecessor problem and related problems.
\newblock {\em J. Comput. Syst. Sci}, 65:38--72, 2002.

\bibitem{BB08}
D.~K. Blandford and G.~E. Blelloch.
\newblock Compact dictionaries for variable-length keys and data with
  applications.
\newblock {\em ACM Transactions on Algorithms}, 4(2):17:1--17:25, May 2008.

\bibitem{BM99}
A.~Brodnik and J.I. Munro.
\newblock Membership in constant time and almost-minimum space.
\newblock {\em SIAM J. Computing}, 28:1627--1640, 1999.

\bibitem{Buhrman:2002:BO}
H.~Buhrman, P.~B. Miltersen, J.~Radhakrishnan, and S.~Venkatesh.
\newblock Are bitvectors optimal?
\newblock {\em SIAM Journal on Computing}, 31(6):1723--1744, December 2002.

\bibitem{CM96}
D.R. Clark and J.I. Munro.
\newblock Efficient suffix trees on secondary storage.
\newblock In {\em Proc. 7th ACM-SIAM SODA}, pages 383--391, 1996.

\bibitem{CN08}
F.~Claude and G.~Navarro.
\newblock Practical rank/select queries over arbitrary sequences.
\newblock In {\em Proc. 15th (SPIRE)}, LNCS. Springer, 2008.

\bibitem{Elias74}
P.~Elias.
\newblock Efficient storage and retrieval by content and address of static
  files.
\newblock {\em J. Assoc. Comput. Mach.}, 21(2):246--260, 1974.

\bibitem{Fano71}
R.~M. Fano.
\newblock On the number of bits required to implement an associative memory.
\newblock {\em Memorandum 61, Computer Structures Group, Project MAC}, 1971.

\bibitem{FGGSV08}
P.~Ferragina, R.~Grossi, A.~Gupta, R.~Shah, and J.~S. Vitter.
\newblock On searching compressed string collections cache-obliviously.
\newblock In {\em Proc. 27th {ACM} {PODS}}, pages 181--190, 2008.

\bibitem{FG99}
P.~Ferragina and Roberto Grossi.
\newblock The string b-tree: A new data structure for string search in external
  memory and its applications.
\newblock {\em J. of the ACM}, 46(2):236--280, 1999.

\bibitem{FLMM05}
P.~Ferragina, F.~Luccio, G.~Manzini, and S.~Muthukrishnan.
\newblock Structuring labeled trees for optimal succinctness, and beyond.
\newblock In {\em Proc. 46th IEEE FOCS}, pages 184--196, 2005.

\bibitem{FM05}
P.~Ferragina and G.~Manzini.
\newblock Indexing compressed text.
\newblock {\em Journal of the ACM}, 52(4):552--581, July 2005.

\bibitem{GM07}
A.~G\'{a}l and P.~Bro Miltersen.
\newblock The cell probe complexity of succinct data structures.
\newblock {\em Theor. Comput. Sci.}, 379:405--417, 2007.

\bibitem{GRRR06}
R.~F. Geary, N.~Rahman, R.~Raman, and V.~Raman.
\newblock A simple optimal representation for balanced parentheses.
\newblock {\em Theor. Comput. Sci.}, 368:231--246, 2006.

\bibitem{GRR06}
R.~F. Geary, R.~Raman, and V.~Raman.
\newblock Succinct ordinal trees with level-ancestor queries.
\newblock {\em ACM Transactions on Algorithms}, 2:510--534, 2006.

\bibitem{golynski07}
A.~Golynski.
\newblock Optimal lower bounds for rank and select indexes.
\newblock {\em Theor. Comput. Sci.}, 387:348--359, 2007.

\bibitem{GGGRR07}
A.~Golynski, R.~Grossi, A.~Gupta, R.~Raman, and S.~S. Rao.
\newblock On the size of succinct indices.
\newblock In {\em Proc 15th ESA, LNCS 4698}, pages 371--382, 2007.

\bibitem{GRR08}
A.~Golynski, R.~Raman, and S.~S. Rao.
\newblock On the redundancy of succinct indices.
\newblock In {\em Proc. 11th SWAT}, pages 148--159, 2008.

\bibitem{GGMN05}
R.~Gonz\'alez, Sz. Grabowski, V.~M{\"akinen}, and G.~Navarro.
\newblock Practical implementation of rank and select queries.
\newblock In {\em Proc. 4th (WEA)}, pages 27--38, 2005.

\bibitem{GGV03}
R.~Grossi, A.~Gupta, and J.~S. Vitter.
\newblock High-order entropy-compressed text indexes.
\newblock In {\em Proc. 14th ACM-SIAM SODA}, pages 841--850, 2003.

\bibitem{GGV04}
R.~Grossi, A.~Gupta, and J.~S. Vitter.
\newblock When indexing equals compression: experiments with compressing suffix
  arrays and applications.
\newblock In {\em Proc. 15th ACM-SIAM SODA}, pages 636--645, 2004.

\bibitem{GV05}
R.~Grossi and J.~S. Vitter.
\newblock Compressed suffix arrays and suffix trees with applications to text
  indexing and string matching.
\newblock {\em SIAM J. Comput}, 35(2):378--407, 2005.

\bibitem{GHSV07}
A.~Gupta, W.~Hon, R.~Shah, and Jeffrey~Scott Vitter.
\newblock Compressed data structures: Dictionaries and data-aware measures.
\newblock {\em Theor. Comput. Sci}, 387(3):313--331, 2007.

\bibitem{HSS03}
W.~Hon, K.~Sadakane, and W.~Sung.
\newblock Breaking a time-and-space barrier in constructing full-text indices.
\newblock In {\em Proc. 44th {IEEE} {FOCS}}, pages 251--260, 2003.

\bibitem{jacobson-thesis}
G.~Jacobson.
\newblock {\em Succinct Static Data Structures}.
\newblock PhD thesis, Carnegie Mellon University, 1989.

\bibitem{M05}
P.~B. Miltersen.
\newblock Lower bounds on the size of selection and rank indexes.
\newblock In {\em Proc. ACM-SIAM SODA}, pages 11--12, 2005.

\bibitem{munro96}
J.~I. Munro.
\newblock Tables.
\newblock In {\em Proc. FST \& TCS, LNCS 1180}, pages 37--42, 1996.

\bibitem{munro08}
J.~I. Munro.
\newblock Lower bounds for succinct data structures.
\newblock In {\em Proc. 19th CPM}, page~3, 2008.

\bibitem{MRRR03}
J~I. Munro, R.~Raman, V.~Raman, and S.~S. Rao.
\newblock Succinct representations of permutations.
\newblock In {\em Proc. 30th ICALP, LNCS 2719}, pages 345--356, 2003.

\bibitem{MR01}
J.~I. Munro and V.~Raman.
\newblock Succinct representation of balanced parentheses and static trees.
\newblock {\em SIAM J. Comput.}, 31:762--776, 2001.

\bibitem{MRR01}
J.~I. Munro, V.~Raman, and S.~S. Rao.
\newblock Space efficient suffix trees.
\newblock {\em J. of Algorithms}, 39:205--222, 2001.

\bibitem{NM07}
G.~Navarro and V.~M{\"a}kinen.
\newblock Compressed full-text indexes.
\newblock {\em ACM Computing Surveys}, 39(1):2:1--2:61, 2007.

\bibitem{OS07}
D.~Okanohara and K.~Sadakane.
\newblock Practical entropy-compressed rank/select dictionary.
\newblock In {\em ALENEX}. SIAM, 2007.

\bibitem{pagh01}
R.~Pagh.
\newblock Low redundancy in static dictionaries with constant query time.
\newblock {\em SIAM J. Computing}, 31:353--363, 2001.

\bibitem{P08}
M.~P\v{a}tra\c{s}cu.
\newblock Succincter.
\newblock In {\em To Appear in Proc. 49th IEEE FOCS}, 2008.

\bibitem{PT06}
M.~P\v{a}tra\c{s}cu and M.~Thorup.
\newblock Time-space trade-offs for predecessor search.
\newblock In {\em Proc. 38th ACM STOC}, pages 232--240, 2006.

\bibitem{RRR07}
R.~Raman, V.~Raman, and S.~S. Rao.
\newblock Succinct indexable dictionaries, with applications to representing
  $k$-ary trees, prefix sums and multisets.
\newblock {\em ACM Transactions on Algorithms}, 3(4), 2007.

\bibitem{MST::BoasKZ1977}
P.~van Emde~Boas, R.~Kaas, and E.~Zijlstra.
\newblock Design and implementation of an efficient priority queue.
\newblock {\em Mathematical Systems Theory}, 10:99--127, 1977.

\bibitem{Vigna08}
S.~Vigna.
\newblock Broadword implementation of rank/select queries.
\newblock In {\em Proc. 7th WEA}, pages 154--168, 2008.

\end{thebibliography}
